\documentclass{article}

\usepackage{fancybox}
\usepackage{amsmath}
\usepackage{amscd}
\usepackage{moreverb}
\usepackage{commath}
\usepackage{algorithm2e}
\usepackage{listings}
\usepackage[standard]{ntheorem}

\usepackage{dsfont}

\usepackage{stmaryrd}

\usepackage{graphicx}
\usepackage{subfigure}

\usepackage{multirow}
\usepackage{array}

\usepackage{color}

\newcommand{\X}{\mathcal{X}}
\newcommand{\Go}{G_{f_0}}
\newcommand{\B}{\mathds{B}}

\begin{document}
%


\title{Lyapunov exponent evaluation of a digital watermarking
scheme proven to be secure}

\author{Jacques M. Bahi, Nicolas Friot, and  Christophe Guyeux*\\
\\
Computer science laboratory DISC\\
FEMTO-ST Institute, UMR 6174 CNRS\\
University of Franche-Comt\'{e}, Besan\c con, France\\
\{jacques.bahi, nicolas.friot, christophe.guyeux\}@femto-st.fr\\
\\
* Authors in alphabetic order
}


\maketitle

\begin{abstract}
In our previous researches, a new digital watermarking
scheme based on chaotic iterations has been introduced. This scheme was both
stego-secure and topologically secure. The stego-security is to face an
attacker in the ``watermark only
attack'' category, whereas the topological security concerns
other categories of attacks. Its Lyapunov
exponent is evaluated here, to quantify the 
chaos generated by this scheme.\newline

\end{abstract}

%
%

\section{Introduction}
It currently exists only three data hiding schemes being both stego-secure and
topologically secure. 
The first one is the ``Natural Watermarking'' 
with
parameter $\eta = 1$~\cite{Cayre2008}. The two others are based on chaotic iterations. The first of them is a one bit watermarking scheme~\cite{gfb10:ip,bcg11b:ip}, 
whereas the last one allows steganographic operations~\cite{fgb11:ip}. 
In order to enlarge the knowledge about the
security of these processes, the Lyapunov exponent of the
digital watermarking scheme based on chaotic iterations is evaluated here.

This document is organized as follows. In Section
\ref{sec:basic-reminders}, some basic reminders are given.
The semiconjugacy allowing the 
exponent evaluation is described in Sect.~\ref{sec:topological-semiconjugacy}. In the next one, the
exponent is evaluated. This paper ends by a conclusion
section where our contribution is summarized.

\section{Basic Reminders}
\label{sec:basic-reminders}

\subsection{Chaotic Iterations and Watermarking Scheme}
\label{sec:chaotic iterations}

Let us consider  a \emph{system} with a finite  number $\mathsf{N} \in
\mathds{N}^*$ of \emph{cells}, so that each  cell has a
boolean  \emph{state}. A sequence which  elements belong into $\llbracket
1;\mathsf{N} \rrbracket $ is a \emph{strategy}. Finally, the set of all strategies is
denoted by $\llbracket 1, \mathsf{N} \rrbracket^\mathds{N}.$ Let $S^{n}$ denotes
the $n^{th}$ term of a sequence $S$, and $V_{i}$ the $i^{th}$ component of a vector $V$.

\begin{definition}
\label{Def:chaotic iterations}
The      set       $\mathds{B}$      denoting      $\{0,1\}$,      let
$f:\mathds{B}^{\mathsf{N}}\longrightarrow  \mathds{B}^{\mathsf{N}}$ be
a  function  and  $S\in  \llbracket 1, \mathsf{N} \rrbracket^\mathds{N}$.  The  
\emph{chaotic      iterations}     are     defined      by     $x^0\in
\mathds{B}^{\mathsf{N}}$ and
\begin{equation*}
\forall    n\in     \mathds{N}^{\ast     },    \forall     i\in
\llbracket1;\mathsf{N}\rrbracket ,x_i^n=\left\{
\begin{array}{ll}
  x_i^{n-1} &  \text{ if  }S^n\neq i \\
  \left(f(x^{n-1})\right)_{S^n} & \text{ if }S^n=i.
\end{array}\right.
\end{equation*}
\end{definition}

In other words, at the $n^{th}$ iteration, only the $S^{n}-$th cell is
\textquotedblleft  iterated\textquotedblright .
Let us now recall how to define a suitable metric space where chaotic iterations
are continuous~\cite{guyeux10}.

Let $\delta $ be the \emph{discrete boolean metric}, $\delta
(x,y)=0\Leftrightarrow x=y.$ Given a function $f$, define the function:
\begin{equation*}
\begin{array}{ll}
F_{f}: & \llbracket1;\mathsf{N}\rrbracket\times \mathds{B}^{\mathsf{N}} 
\longrightarrow  \mathds{B}^{\mathsf{N}} \\
& (k,E)  \longmapsto  \left( E_{j}.\delta (k,j)+f(E)_{k}.\overline{\delta
(k,j)}\right) _{j\in \llbracket1;\mathsf{N}\rrbracket}%
\end{array}%
\end{equation*}%
\noindent Consider the phase space
$\mathcal{X} = \llbracket 1 ; \mathsf{N} \rrbracket^\mathds{N} \times
\mathds{B}^\mathsf{N}$,
\noindent and the map defined on $\mathcal{X}$ by:
\begin{equation}
G_f\left(S,E\right) = \left(\sigma(S), F_f(i(S),E)\right), \label{Gf}
\end{equation}
\noindent where $\sigma :
(S^{n})_{n\in \mathds{N}}\in \llbracket 1, \mathsf{N} \rrbracket^\mathds{N}\longrightarrow (S^{n+1})_{n\in
\mathds{N}} \in \llbracket 1, \mathsf{N} \rrbracket^\mathds{N}$ and 
$i:(S^{n})_{n\in \mathds{N}} \in \llbracket 1,
\mathsf{N} \rrbracket^\mathds{N}\longrightarrow S^{0} \in \llbracket 1,
\mathsf{N} \rrbracket$ are respectively the \emph{shift} and the \emph{initial}
functions. Then chaotic iterations can be described by the following discreet
dynamical system:
\begin{equation}
\left\{
\begin{array}{l}
X^0 \in \mathcal{X} \\
X^{k+1}=G_{f}(X^k).%
\end{array}%
\right.
\end{equation}%

To study whether this dynamical system is chaotic~\cite{Devaney}, a distance
between  $X = (S,E)$ and $Y = (\check{S},\check{E})\in
\mathcal{X}$ has been introduced in~\cite{guyeux10} as follows:
$d(X,Y)=d_{e}(E,\check{E})+d_{s}(S,\check{S})$, where: 
\begin{equation*}
\displaystyle{d_{e}(E,\check{E})}  =  \displaystyle{\sum_{k=1}^{\mathsf{N}%
}\delta (E_{k},\check{E}_{k})} \text{ and }\displaystyle{d_{s}(S,\check{S})}  = 
\displaystyle{\dfrac{9}{\mathsf{N}}\sum_{k=1}^{\infty
}\dfrac{|S^k-\check{S}^k|}{10^{k}}}.
\end{equation*}

This distance has been introduced to satisfy the following requirements. If the floor
value $\lfloor d(X,Y)\rfloor $ is equal to $n$, then the states $E$ and
$\check{E}$ differ in $n$ cells. In addition, its floating part is less than $10^{-k}$ if and only if the first $k$ terms of the two strategies are equal. Moreover, if the $k^{th}$ digit is
nonzero, then the $k^{th}$ terms of the two strategies are different.
With this metric, and the boolean vectorial negation $f_0$,
it has been
proven in~\cite{guyeux10} that,

\begin{theorem}
$G_{f_{0}}$ is continuous and chaotic in $(\mathcal{X},d)$.
\end{theorem}

The digital watermarking scheme proposed in~\cite{gfb10:ip,bcg11b:ip} 
is simply the iterations of this dynamical system on the least significant 
coefficients of the considered media. Each property exhibited by 
the dynamical system will then be possessed too by the watermarking scheme.
For further explanations, see~\cite{gfb10:ip,bcg11b:ip}.

\subsection{The Lyapunov Exponent}

Some dynamical systems are very sensitive to small changes in their initial
condition, which is illustrated by both the constants of sensitivity to initial
conditions and of expansivity~\cite{guyeux10}. However, these
variations can quickly take enormous proportions, grow exponentially, and none
of these constants can illustrate that. Alexander Lyapunov has examined this
phenomenon and introduced an exponent that measures the rate at which these
small variations can grow:

\begin{definition}
\label{def:lyapunov}
Given $f: \mathds{R} \longrightarrow \mathds{R}$, the \emph{Lyapunov exponent}
of the system composed by $x^0 \in \mathds{R}$ and $x^{n+1} = f(x^n)$
\noindent is defined by
$\displaystyle{\lambda(x_0)=\lim_{n \to +\infty} \dfrac{1}{n} \sum_{i=1}^n \ln \left| ~f'\left(x^{i-1}\right)\right|}$.
\end{definition}

Consider a dynamic system with an infinitesimal error on the initial condition $x_0$. 
When the Lyapunov exponent is positive, this
error will increase (situation of chaos), whereas it will decrease if $\lambda(x_0)\leqslant 0$.

\begin{example}
The Lyapunov exponent of the logistic map~\cite{Arroyo08} becomes positive for
$\mu>3,54$, but it is always smaller than 1. The tent map~\cite{Wang20093089} 
and the doubling map of the circle~\cite{Richeson2008251} have a Lyapunov exponent equal to $\ln(2)$.
\end{example}

To evaluate the Lyapunov exponent of our digital watermarking scheme,
 chaotic iterations must be
described by a differentiable function on $\mathds{R}$. To do so,
a topological semiconjugacy between the phase space $\mathcal{X}$ and
$\mathds{R}$ must be written.

\section{A Topological Semiconjugacy}
\label{sec:topological-semiconjugacy}
\subsection{The Phase Space is an Interval of the Real Line}

\subsubsection{Toward a Topological Semiconjugacy}

We show, by using a topological
semiconjugacy, that chaotic iterations on $\mathcal{X}$ can be
described as
iterations on a real interval. To do so, some
notations and terminologies must be introduced. 

Let $\mathcal{S}_\mathsf{N}=\llbracket 1 ; \mathsf{N} \rrbracket^\mathds{N}$ be
the set of sequences belonging into $\llbracket
1; \mathsf{N}\rrbracket$ and $\mathcal{X}_{\mathsf{N}} = \mathcal{S}_
\mathsf{N}
\times \B^\mathsf{N}$.
In what follows and for easy understanding, 
we will assume that $N = 10$.
However, an 
equivalent formulation of the following can be easily obtained 
by replacing the base $10$ by any base $\mathsf{N}$.

\begin{definition}
The function $\varphi: \mathcal{S}_{10} \times\mathds{B}^{10}
\rightarrow \big[
0, 2^{10} \big[$ is defined by:
\begin{equation*}
 \begin{array}{cccl}
\varphi: & \mathcal{X}_{10} = \mathcal{S}_{10} \times\mathds{B}^{10}&
\longrightarrow & \big[ 0, 2^{10} \big[ \\
 & \left((S^0, S^1, \hdots ); (E_0, \hdots, E_9)\right) &
\longmapsto &
\varphi \left((S,E)\right)
\end{array}
\end{equation*}
where $(S,E) = \left((S^0, S^1, \hdots ); (E_0, \hdots, E_9)\right)$, and
$\varphi\left((S,E)\right)$ is the real number:
\begin{itemize}
\item whose integral part $e$ is $\displaystyle{\sum_{k=0}^9 2^{9-k}
E_k}$, that
is, the binary digits of $e$ are $E_0 ~ E_1 ~ \hdots ~ E_9$.
\item whose decimal part $s$ is equal to $s = 0,S^0~ S^1~ S^2~ \hdots =
\sum_{k=1}^{+\infty} 10^{-k} S^{k-1}.$ 
\end{itemize}
\end{definition}

$\varphi$ realizes the association between a point of $\mathcal{X}_{10}$
and a
real number into $\big[ 0, 2^{10} \big[$. We must now translate the digital
watermarking process $\Go$ based  on chaotic
iterations  on this real interval. To do so, two intermediate
functions over $\big[ 0, 2^{10} \big[$ denoted by $e$ and $s$ must be
introduced:

\begin{definition}
\label{def:e et s}
Let $x \in \big[ 0, 2^{10} \big[$ and:
\begin{itemize}
\item $e_0, \hdots, e_9$ the binary digits of the integral part of $x$:
$\displaystyle{\lfloor x \rfloor = \sum_{k=0}^{9} 2^{9-k} e_k}$.
\item $(s^k)_{k\in \mathds{N}}$ the digits of $x$, where the chosen
decimal
decomposition of $x$ is the one that does not have an infinite number of
9: 
$\displaystyle{x = \lfloor x \rfloor + \sum_{k=0}^{+\infty} s^k
10^{-k-1}}$.
\end{itemize}
$e$ and $s$ are thus defined as follows:
\begin{equation*}
\begin{array}{cccl}
e: & \big[ 0, 2^{10} \big[ & \longrightarrow & \mathds{B}^{10} \\
 & x & \longmapsto & (e_0, \hdots, e_9)
\end{array}
\end{equation*}
and
\begin{equation*}
 \begin{array}{cccc}
s: & \big[ 0, 2^{10} \big[ & \longrightarrow & \llbracket 0, 9
\rrbracket^{\mathds{N}} \\
 & x & \longmapsto & (s^k)_{k \in \mathds{N}}
\end{array}
\end{equation*}
\end{definition}

We are now able to define the function $g$, whose goal is to translate
the
chaotic iterations $\Go$ on an interval of $\mathds{R}$.

\begin{definition}\label{def:function-g-on-R}
$g:\big[ 0, 2^{10} \big[ \longrightarrow \big[ 0, 2^{10} \big[$ is
by definition such that g(x) is the real number of $\big[ 0, 2^{10} \big[$
defined bellow:
\begin{itemize}
\item its integral part has a binary decomposition equal to $e_0',
\hdots,
e_9'$, with:
 \begin{equation*}
e_i' = \left\{
\begin{array}{ll}
e(x)_i & \textrm{ if } i \neq s^0\\
e(x)_i + 1 \textrm{ (mod 2)} & \textrm{ if } i = s^0\\
\end{array}
\right.
\end{equation*}
\item whose decimal part is $s(x)^1, s(x)^2, \hdots$
\end{itemize}
\end{definition}

In other words, if $x = \displaystyle{\sum_{k=0}^{9} 2^{9-k} e_k + 
\sum_{k=0}^{+\infty} s^{k} ~10^{-k-1}}$, then:
\begin{equation*}
g(x) =
\displaystyle{\sum_{k=0}^{9} 2^{9-k} (e_k + \delta(k,s^0) \textrm{ (mod
2)}) + 
\sum_{k=0}^{+\infty} s^{k+1} 10^{-k-1}}. 
\end{equation*}

\subsubsection{Defining a Metric on $\big[ 0, 2^{10} \big[$}

Numerous metrics can be defined on the set $\big[ 0, 2^{10} \big[$, the
most
usual one being the Euclidian distance
$\Delta(x,y) = |y-x|^2$.
This Euclidian distance does not reproduce exactly the notion of
proximity
induced by our first distance $d$ on $\X$. Indeed $d$ is richer than
$\Delta$.
This is the reason why we have to introduce the following metric:

\begin{definition}
Given $x,y \in \big[ 0, 2^{10} \big[$,
$D$ denotes the function from $\big[ 0, 2^{10} \big[^2$ to $\mathds{R}^
+$
defined by: $D(x,y) = D_e\left(e(x),e(y)\right) + D_s
\left(s(x),s(y)\right)$,
where:
\begin{center}
$\displaystyle{D_e(e,\check{e}) = \sum_{k=0}^\mathsf{9} \delta (e_k,
\check{e}_k)}$, ~~and~ $\displaystyle{D_s(s,\check{s}) = \sum_{k = 1}^
\infty
\dfrac{|s^k-\check{s}^k|}{10^k}}$.
\end{center}
\end{definition}

\begin{proposition}
$D$ is a distance on $\big[ 0, 2^{10} \big[$.
\end{proposition}

\begin{proof}
The three axioms defining a distance must be checked.
\begin{itemize}
\item $D \geqslant 0$, because everything is positive in its definition.
If
$D(x,y)=0$, then $D_e(x,y)=0$, so the integral parts of $x$ and $y$ are
equal
(they have the same binary decomposition). Additionally, $D_s(x,y) = 0$,
then
$\forall k \in \mathds{N}^*, s(x)^k = s(y)^k$. In other words, $x$ and
$y$ have
the same $k-$th decimal digit, $\forall k \in \mathds{N}^*$. And so $x=y
$.
\item $D(x,y)=D(y,x)$.
\item Finally, the triangular inequality is obtained due to the fact
that both
$\delta$ and $|x-y|$ satisfy it.
\end{itemize}
\end{proof}

The convergence of sequences according to $D$ is not the same than the
usual
convergence related to the Euclidian metric. For instance, if $x^n \to x
$
according to $D$, then necessarily the integral part of each $x^n$ is
equal to
the integral part of $x$ (at least after a given threshold), and the
decimal
part of $x^n$ corresponds to the one of $x$ ``as far as required''.
To illustrate this fact, a comparison between $D$ and the Euclidian
distance is
given Figure \ref{fig:comparaison de distances}. These illustrations
show that
$D$ is richer and more refined than the Euclidian distance, and thus is
more
precise.

\begin{figure}[t]
\begin{center}
  \subfigure[Function $x \to dist(x;1,234) $ on the interval
$(0;5)$.]{\includegraphics[scale=.4]{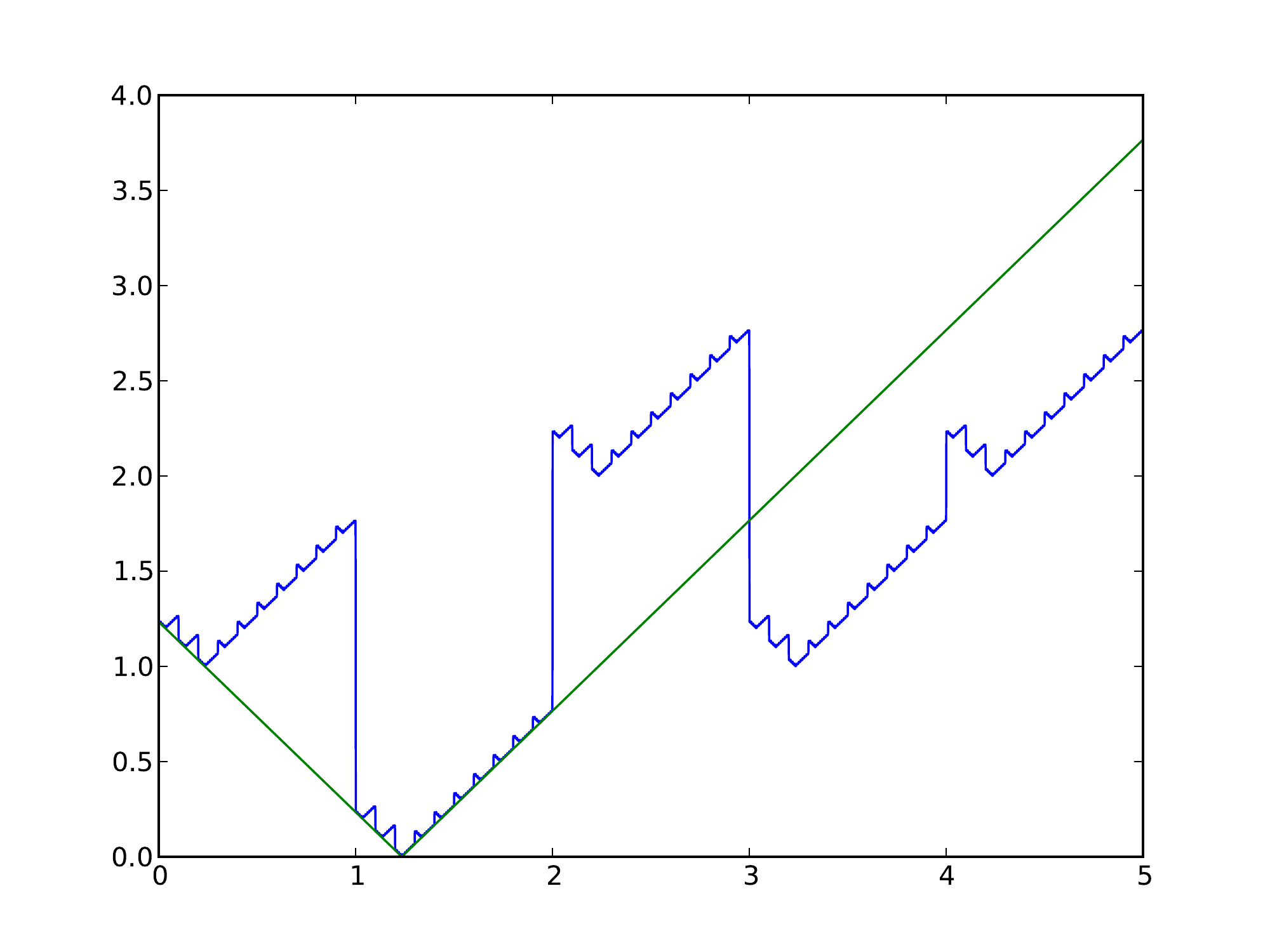}}\quad
  \subfigure[Function $x \to dist(x;3) $ on the interval
$(0;5)$.]{\includegraphics[scale=.4]{DvsEuclidien2}}
\end{center}
\caption{Comparison between $D$ (in blue) and the Euclidian distance (in
green).}
\label{fig:comparaison de distances}
\end{figure}

\subsubsection{The Semiconjugacy}

It is now possible to define a topological semiconjugacy between
$\mathcal{X}$
and an interval of $\mathds{R}$:

\begin{theorem}
Chaotic iterations on the phase space $\mathcal{X}$ are simple
iterations on
$\mathds{R}$, which is illustrated by the semiconjugacy given
bellow:
\begin{equation*}
\begin{CD}
\left(~\mathcal{S}_{10} \times\mathds{B}^{10}, d~\right) @>G_{f_0}>>
\left(~\mathcal{S}_{10} \times\mathds{B}^{10}, d~\right)\\
    @V{\varphi}VV                    @VV{\varphi}V\\
\left( ~\big[ 0, 2^{10} \big[, D~\right)  @>>g> \left(~\big[ 0, 2^{10}
\big[,
D~\right)
\end{CD}
\end{equation*}
\end{theorem}

\begin{proof}
$\varphi$ has been constructed in order to be continuous and onto.
\end{proof}

In other words, $\mathcal{X}$ is approximately equal to $\big[ 0, 2^
\mathsf{N}
\big[$.

\subsection{Chaotic Iterations Described as a Real
Function}

%

%



It can be remarked that the function $g$ is a piecewise linear function:
it is
linear on each interval having the form $\left[ \dfrac{n}{10},
\dfrac{n+1}{10}\right[$, $n \in \llbracket 0;2^{10}\times 10 \rrbracket$
and its
slope is equal to 10. Let us justify these claims:

\begin{proposition}
\label{Prop:derivabilite des ICs}
Chaotic iterations $g$ defined on $\mathds{R}$ have derivatives of all
orders on
$\big[ 0, 2^{10} \big[$, except on the 10241 points in $I$ defined by
$\left\{
\dfrac{n}{10} ~\big/~ n \in \llbracket 0;2^{10}\times 10\rrbracket
\right\}$.

Furthermore, on each interval of the form $\left[ \dfrac{n}{10},
\dfrac{n+1}{10}\right[$, with $n \in \llbracket 0;2^{10}\times 10
\rrbracket$,
$g$ is a linear function, having a slope equal to 10: $\forall x \notin
I,
g'(x)=10$.
\end{proposition}

\begin{proof}
Let $I_n = \left[ \dfrac{n}{10}, \dfrac{n+1}{10}\right[$, with $n \in
\llbracket
0;2^{10}\times 10 \rrbracket$. All the points of $I_n$ have the same
integral
part $e$ and the same decimal part $s^0$: on the set $I_n$,  functions
$e(x)$
and $x \mapsto s(x)^0$ of Definition \ref{def:e et s} only depend on $n
$. So all
the images $g(x)$ of these points $x$:
\begin{itemize}
\item Have the same integral part, which is $e$, except probably the bit
number
$s^0$. In other words, this integer has approximately the same binary
decomposition than $e$, the sole exception being the digit $s^0$ (this
number is
then either $e+2^{10-s^0}$ or $e-2^{10-s^0}$, depending on the parity of
$s^0$,
\emph{i.e.}, it is equal to $e+(-1)^{s^0}\times 2^{10-s^0}$).
\item A shift to the left has been applied to the decimal part $y$,
losing by
doing so the common first digit $s^0$. In other words, $y$ has been
mapped into
$10\times y - s^0$.
\end{itemize}
To sum up, the action of $g$ on the points of $I$ is as follows: first,
make a
multiplication by 10, and second, add the same constant to each term,
which is
$\dfrac{1}{10}\left(e+(-1)^{s^0}\times 2^{10-s^0}\right)-s^0$.
\end{proof}

\begin{remark}
Finally, chaotic iterations used in our watermarking scheme are elements of the large family of
functions that
are both chaotic and piecewise linear (like the tent map~\cite{Wang20093089}).
\end{remark}

We are now able to evaluate the Lyapunov exponent of our digital watermarking
scheme based on chaotic iterations, which is now described by the iterations on $\mathds{R}$
of the $g$ function introduced in Definition~\ref{def:function-g-on-R}.

\section{Evaluation of the Lyapunov Exponent}



Let $\mathcal{L} = \left\{ x^0 \in \big[ 0, 2^{10} \big[ ~ \big/ ~ \forall n
\in \mathds{N}, x^n \notin I \right\}$, where $I$ is the set of points in the
real interval where $g$ is not differentiable (as it is explained in Proposition
\ref{Prop:derivabilite des ICs}). Then,

\begin{theorem}
$\forall x^0 \in \mathcal{L}$, the Lyapunov exponent of chaotic iterations having $x^0$ for initial condition is equal to
$\lambda(x^0) = \ln (10)$.
\end{theorem}

\begin{proof}
It is reminded that $g$ is piecewise linear, with a slop of 10 ($g'(x)=10$
where the function $g$ is differentiable). Then $\forall x \in
\mathcal{L}$, $\lambda (x) = \lim_{n \to +\infty} \dfrac{1}{n} \sum_{i=1}^n
 \ln \left| ~g'\left(x^{i-1}\right)\right|
=  \lim_{n \to +\infty} \dfrac{1}{n} \sum_{i=1}^n \ln \left|10\right| 
= \lim_{n\to +\infty} \dfrac{1}{n} n \ln \left|10\right| = \ln 10.$
\end{proof}

\begin{remark}
The set of initial conditions for which this exponent is not calculable is countable.
This is indeed the initial conditions such that an iteration value will be a
number having the form $\dfrac{n}{10}$, with $n\in \mathds{N}$. We can reach such a
 real number only by starting iterations on a \emph{decimal number}, as this latter must have
a finite fractional part.
\end{remark}

\begin{remark}
For a system having $\mathsf{N}$
cells, we will find, mutatis mutandis, an infinite uncountable set of initial conditions $x^0 \in \left[0;2^\mathsf{N}\right[$ such that
$\lambda (x^0) = \ln (\mathsf{N})$.
\end{remark}

So, it is possible to make the Lyapunov exponent of our digital watermarking scheme as large as possible, depending on the number of least significant coefficients of the cover media we decide to consider.
Obviously, a large Lyapunov exponent make it impossible to achieve the well-known Original Estimated
Attacks~\cite{Cayre2008}.


\section{Conclusion and Future Works}

As a conclusion, we have available to us now  a new quantitative property concerning our
digital watermarking scheme based on chaotic iteration: its Lyapunov exponent
 is equal to $\ln(\mathsf{N})$, where $\mathsf{N}$ is the number of 
least significant coefficients of the cover media. This exponent allows to quantify the amplification
of the ignorance on the exact initial condition (the media without watermark) after several iterations of the
watermaking process. It illustrates the disorder
generated by iterations of our watermarking process, reinforcing its chaotic nature.

Using the semiconjugacy described here, it will be possible in a
future work to compare the topological behavior of chaotic iterations on
$\mathcal{X}$ and  $\mathds{R}$, and to explore the topological security of
the watermarking scheme 
using this new topology. Finally, an analogue study of the two other
topologically secure schemes will be also conducted in order to compare these processes, being thus able to choose the best one according to the type of applications under consideration.

\bibliographystyle{plain}
\bibliography{jabref}

\end{document}